\documentclass[journal,onecolumn,draft]{IEEEtran}
\usepackage{geometry}                		
\usepackage[final]{graphicx}				
\usepackage{amsmath}								
\usepackage{amssymb}
\usepackage{bm}
\usepackage{amsthm}

\usepackage{graphicx}
\usepackage{caption}
\usepackage{subcaption}

\newtheorem{proposition}{Proposition}
\newtheorem{definition}{Definition}
\newtheorem{lemma}{Lemma}
\newtheorem{theorem}{Theorem}
\newtheorem{corollary}{Corollary}
\newtheorem{example}{Example}

\makeatletter
\newcommand*{\centerfloat}{%
	\parindent \z@
	\leftskip \z@ \@plus 1fil \@minus \textwidth
	\rightskip\leftskip
	\parfillskip \z@skip}
\makeatother

\ifCLASSINFOpdf
\else
\fi
\begin{document}
%
\title{Communication Efficient Secret Sharing}
%
%
%

\author{Wentao~Huang,
        Michael~Langberg,~\IEEEmembership{senior member,~IEEE,}
        Joerg~Kliewer,~\IEEEmembership{senior member,~IEEE,}
        and~Jehoshua~Bruck,~\IEEEmembership{Fellow,~IEEE}
\thanks{W. Huang and J. Bruck are with the Department
of Electrical Engineering, California Institute of Technology, Pasadena,
CA, 91125 USA (e-mail:\{whuang,bruck\}@caltech.edu).}
\thanks{M. Langberg is with the Department of Electrical Engineering, The State University of New York at Buffalo, Buffalo, NY 14260 USA (email: mikel@buffalo.edu)}
\thanks{J. Kliewer is with the Department of Electrical and Computer Engineering, New Jersey Institute of Technology, Newark, NJ 07102 USA (email: jkliewer@njit.edu)}}

\maketitle

\begin{abstract}
A secret sharing scheme is a method to store information securely and reliably. Particularly, in a \emph{threshold secret sharing scheme}, a secret is encoded into $n$ shares, such that any set of at least $t_1$ shares suffice to decode the secret, and any set of at most $t_2 < t_1$ shares reveal no information about the secret. Assuming that each party holds a share and a user wishes to decode the secret by receiving information from a set of parties; the question we study is how to minimize the amount of communication between the user and the parties. We show that the necessary amount of communication, termed ``decoding bandwidth'', decreases as the number of parties that participate in decoding increases. We prove a tight lower bound on the decoding bandwidth, and construct secret sharing schemes achieving the bound. Particularly, we design a scheme that achieves the optimal decoding bandwidth when $d$ parties participate in decoding, universally for all $t_1 \le d \le n$. The scheme is based on Shamir's secret sharing scheme and preserves its simplicity and efficiency. In addition, we consider secure distributed storage where the proposed communication efficient secret sharing schemes further improve disk access complexity during decoding.
\end{abstract}


%
\IEEEpeerreviewmaketitle

\section{Introduction}
Consider the scenario that $n$ parties wish to store a secret securely and reliably. 
To this end, a dealer distributes the secret into $n$ shares, i.e., one share for each party, such that 1) (reliability) a collection $\mathcal{A}$ of ``authorized'' subsets of the parties can decode the secret,  and 2) (secrecy) 
a collection $\mathcal{B}$ of ``blocked'' subsets of the parties cannot collude to deduce any information about the secret.  A scheme to distribute the secret into shares with respect to \emph{access structure} $(\mathcal{A},\mathcal{B})$ is called a secret sharing scheme, initially studied in the seminal works by Shamir \cite{Shamir:1979vo} and Blakley \cite{Blakley:1979th}. A secret sharing scheme is \emph{perfect} if a subset of parties is either authorized or blocked, i.e., $\mathcal{A} \cup \mathcal{B} = 2^{\{1,...,n\}}$. The scheme is referred to as a \emph{ramp} scheme if it is not prefect. 
 Besides its application in distributed storage of secret data, secret sharing became a fundamental cryptographic primitive and is used as a building block in numerous secure protocols \cite{Beimel:2011fo}.

We focus on secret sharing schemes for the \emph{threshold} access structure, i.e., $\mathcal{A}$ contains all subsets of $\{1,...,n\}$ of size at least $n-r$, and $\mathcal{B}$ contains all subsets of $\{1,...,n\}$ of size at most $z$.  In other words, the secret can be decoded in the absence of any $r$ parties, and any $z$ parties cannot collude to deduce any information about the secret.
The threshold access structure is particularly important in practice, because for this case, space and computationally efficient secret sharing schemes are known. Specifically, Shamir \cite{Shamir:1979vo} constructs an elegant and efficient perfect threshold scheme using the idea of polynomial interpolation. Shamir's scheme is later shown to be closely related to Reed-Solomon codes \cite{McEliece:1981uy} and is generalized to ramp schemes in \cite{Blakley84, Yamamoto86}, which have significantly better space efficiency, i.e., rate, than the original perfect scheme. Shamir's scheme and the generalized ramp schemes achieve optimal usage of storage space, in the sense that fixing the size of the shares, the schemes store a secret of maximum size. The schemes are computationally efficient as decoding the secret is equivalent to polynomial interpolation. An example of Shamir's ramp scheme is shown in Figure~\ref{fig:shaone}. Other threshold secret sharing schemes and generalizations of Shamir's scheme may be found in \cite{Karnin:1983ft, Yang:2004jd, Lai:2004wy, Kurihara08}. The reader is also referred to \cite{Beimel:2011fo} for an up-to-date survey on secret sharing.
 \begin{figure}[htbp]
 	\small
        \centerfloat
        		\begin{tabular}{c|c|c|c|c|c|c}
		Party 1 & Party 2 & Party 3 & Party 4 & Party 5 & Party 6 & Party 7  \\
		\hline
		$f(1)=$ & $f(2)=$ & $f(3)=$ & $f(4)=$ & $f(5)=$ & $f(6)=$ & $f(7)=$\\
		$m_1+m_2+k$ &  $m_1 + 2m_2 + 4k$ &  $m_1 + 3m_2 + 9k$ & $m_1 + 4m_2 + 5k$ & $m_1+5m_2+3k$ & $m_1 + 6m_2 +3k$ & $m_1+7m_2+5k$ 
		\end{tabular}
                \caption{Shamir's scheme (ramp version) for $n=7, r=4, z=1$, with symbols over $\mathbb{F}_{11}$.  The scheme stores a secret of two symbols, denoted by $m_1, m_2$. Let $k$ be a uniformly and independently distributed random variable. $f(x)$ is the polynomial $m_{1} + m_{2}x + kx^2$. Note that the share stored by any single party is independent of the secret because it is padded by $k$, and that the secret can be decoded from the shares stored by any three parties by polynomial interpolation. }
                \label{fig:shaone}
\end{figure}
 

In addition to space and computational efficiency, this paper studies the \emph{communication efficiency} for secret sharing schemes. Consider the scenario that a user wishes to decode the secret by downloading information from the parties that are available. Referring to the amount of information downloaded by the user as the \emph{decoding bandwidth}, a natural question is to address the minimum decoding bandwidth that allows decoding. It is of practical interest to design secret sharing schemes that achieve a small decoding bandwidth, or in other words, that require communicating only a small amount of information during decoding. In such a case, decoding will be completed in a timely manner and  the communication resource will be more efficiently utilized. 

In many existing  secret sharing schemes, e.g., \cite{Shamir:1979vo, McEliece:1981uy, Karnin:1983ft, Blakley84, Yamamoto86, Yang:2004jd, Lai:2004wy, Kurihara08}, a common practice in decoding is that the user will communicate with a minimum set of parties, i.e., exactly $n-r$ parties (even if $d > n-r$ parties are available) and download the whole share stored by these parties. Wang and Wong \cite{Wang08} show that this paradigm is not optimal in terms of communication and that the decoding bandwidth can be reduced if the user downloads only part of the share from each of the $d>n-r$ available parties. Specifically, given $d$, for any perfect threshold secret sharing scheme, \cite{Wang08} derive a lower bound on the decoding bandwidth when exactly $d$ parties participate in decoding, and design a perfect scheme that achieves the lower bound. The field size of the scheme is slightly improved in \cite{Zhang2012}. However, two interesting and important problems remain open: 1) the schemes in \cite{Wang08,Zhang2012} achieve the lower bound on decoding bandwidth when the number of available parties $d$ equals a single specific value, and do not achieve the bound if $d$ takes other values. This raises the question  whether the lower bound is uniformly tight, or in other words, it is possible to design a single scheme that achieves the lower bound universally for all $d$ in the range of $[n-r, n]$. 2) The results in \cite{Wang08,Zhang2012} target the case of prefect secret sharing schemes. It is well known that for any perfect scheme, the size of each share is as large as the size of the secret \cite{Stinson92, Capocelli93}, i.e., the rate of a perfect scheme is at most $1/n$. Any scheme with a higher rate is necessarily a (non-perfect) ramp scheme, which raises the question of how to generalize the results and ideas to non-perfect schemes.
 Both problems are of practical importance as the first problem addresses the flexibility of a scheme in terms of decoding, and the second problem addresses the high-rate case which is a typical requirement in many practical applications. In this paper we settle both problems and construct (perfect and ramp) schemes of flexible rate that achieve the optimal decoding bandwidth universally. Similar to Shamir's scheme, our schemes are computationally efficient and have optimal space efficiency. 


\subsection{Motivating Example}\label{sec:mex}
\begin{figure}[!ht]
        \centering
        \begin{subfigure}[b]{1\textwidth}
 	\small
 	\centerfloat
 	\begin{tabular}{c|c|c|c|c|c|c}
 		Party 1 & Party 2 & Party 3 & Party 4 & Party 5 & Party 6 & Party 7  \\
 		\hline
 		$m_1+m_2+k_1$ &  $m_1 + 2m_2 + 4k_1$ &  $m_1 + 3m_2 + 9k_1$ & $m_1 + 4m_2 + 5k_1$ & $m_1+5m_2+3k_1$ & $m_1 + 6m_2 +3k_1$ & $m_1+7m_2+5k_1$  \\
 		$m_3+m_4+k_2$ &  $m_3 + 2m_4 + 4k_2$ &  $m_3 + 3m_4 + 9k_2$ & $m_3 + 4m_4 + 5k_2$ & $m_3+5m_4+3k_2$ & $m_3 + 6m_4 +3k_2$ & $m_3+7m_4+5k_2$  \\
 		$m_5+m_6+k_3$ &  $m_5 + 2m_6 + 4k_3$ &  $m_5 + 3m_6 + 9k_3$ & $m_5 + 4m_6 + 5k_3$ & $m_5+5m_6+3k_3$ & $m_5 + 6m_6 +3k_3$ & $m_5+7m_6+5k_3$  
 	\end{tabular}
                \caption{Shamir's Scheme}
                \label{fig:sha}
        \end{subfigure}%
        \\
         \vspace{5mm}
        \begin{subfigure}[b]{1\textwidth}
 	\small
 	\centerfloat
 	\begin{tabular}{c|c|c}
 		Party 1  & $\cdots$ & Party 7  \\
 		\hline
 		$f(1)=k_1+m_1+m_2+m_3+m_4+m_5+m_6$  & $\cdots$ & $f(7) = k_1 + 7m_1 + 5m_2 + 2m_3 + 3m_4 + 10m_5 + 6m_6$ \\
 		$g(1)=k_2+m_4+m_5+m_6$ & $\cdots$ & $g(7)=k_2 + 7m_4 + 5m_5+ 2m_6$ \\
 		$h(1)= k_3 + m_3 + m_6$  & $\cdots$ & $h(7) = k_3 + 7m_3 + 5m_6$ 
 	\end{tabular}
                \caption{Proposed Scheme}
                \label{fig:ours}
        \end{subfigure}
        \caption{ \small Two secret sharing schemes for $n=7, r=4$ and $z=1$ over $\mathbb{F}_{11}$. Both schemes store a secret of six symbols ($m_1$, ..., $m_6$). In both schemes, $k_1, k_2, k_3$ are i.i.d. uniformly distributed random variables. Scheme (a) is Shamir's scheme (see Figure \ref{fig:shaone}) repeated three times. In scheme (b), $f(x)=k_1 + m_1 x +m_2 x^2 +m_3 x^3 + m_4 x^4 + m_5 x^5  + m_6 x^6$, $g(x) = k_2 + m_4 x + m_5 x^2 + m_6 x^3$, $h(x) = k_3 + m_3 x + m_6 x^2$, and party $i$ stores evaluations $f(i)$, $g(i)$ and $h(i)$. Note that in (b), if all 7 parties are available, then the secret can be decoded by downloading only one symbol $f(i)$ from each party $i$, and then interpolating $f(x)$. If any 4 parties are available, then the secret can be decoded in the following way. Download two symbols $f(i), g(i)$ from each available party $i$ and first interpolate $g(x)$, implying that all coefficients of $f(x)$ of degree larger than 3 are decoded. The remaining unknown part of $f(x)$ is a degree-3 polynomial and so we have enough evaluations of $f(x)$ to interpolate it, hence completely decoding the secret. Similarly, if any 3 parties are available, then the secret can be decoded in the following way. Download all three symbols $f(i), g(i), h(i)$ from each available node $i$ and  interpolate $h(x)$, which decodes the degree-3 coefficients of $f(x)$ and $g(x)$. Hence the remaining unknown part of $g(x)$ is a degree-2 polynomial and can be interpolated, which decodes the coefficients of $f(x)$ of degrees $4,5,6$. Hence the remaining unknown part of $f(x)$ is a degree-2 polynomial and can be interpolated, decoding the complete secret. This  shows that the scheme meets the reliability requirement. In fact, for $d=3,4,7$, scheme (b) achieves the optimal decoding bandwidth when $d$ parties participate in decoding. The secrecy of the scheme derives from the secrecy of Shamir's scheme, as each polynomials $f(x)$, $g(x)$ and $h(x)$ individually is an instance of Shamir's scheme, and we show that combining them still meets the secrecy requirement.  The construction is discussed in detail in Section \ref{sec:shar}.
}\label{fig:comp}
\end{figure}

Consider Shamir's ramp scheme in the example of Figure \ref{fig:shaone}, that stores 2 symbols securely and reliably for the setting $n=7, r=4$ and $z=1$. In order to decode the secret, a user needs to download 3 symbols from any 3 parties, 
and therefore the decoding bandwidth is 3 symbols. Now suppose the same scheme is repeated 3 times in order to store a secret of 6 symbols, as shown in Figure \ref{fig:sha}. Then to decode the secret, the decoding bandwidth is 9 symbols. 

We propose  a new scheme in Figure \ref{fig:ours} that also stores a secret of 6 symbols for the same setting, using the same amount of storage space, and over the same field size. In this scheme, if any 3 parties are available, then similar to Shamir's scheme, the secret can be decoded from the 9 symbols stored by the three parties. However, if any 4 parties are available, then the secret can be decoded by downloading 2 symbols from each available party. Therefore, the decoding bandwidth is improved to 8 symbols. If all 7 parties are available,  then the secret can be decoded by downloading only 1 symbol from each party and so the decoding bandwidth is further reduced to 7 symbols.

We use the examples in Figure \ref{fig:comp} to highlight several ideas to reduce the decoding bandwidth. Firstly, the amount of communication depends on the number of available parties. In fact the necessary amount of communication decreases strictly as the number of available parties increases.  Secondly, it is important to distribute multiple subshares (symbols) to a party (essentially using the ideas of array codes \cite{Blaum:1995uv, Blaum:1996ga}). In contrast, Shamir's scheme only distributes one symbol to each party except for trivial repetitions.
Thirdly, during decoding it is not always necessary to download the complete share stored by a party. In general, a party can preprocess its share and the user can download a function of the share.

Comparing to the schemes in \cite{Wang08,Zhang2012}, the scheme in the example is improved and generalized in the following aspects. 1) The proposed scheme achieves the optimal bandwidth more flexibly. Specifically, the schemes in \cite{Wang08,Zhang2012} achieve the optimal bandwidth for a single specific number of available parties. The proposed scheme is more flexible as it can be designed to allow flexibility in the number of available parties $d$.
In the example of Figure \ref{fig:ours} the scheme achieves the optimal bandwidth when $d=3,4,7$. In general, we can construct schemes that achieve the optimal bandwidth for all $n-r \le d \le n$. 2) The proposed scheme is more flexible in rate. 
Specifically, the (perfect) schemes in \cite{Wang08,Zhang2012} have rate exactly $1/n$. The proposed scheme in the example has rate $2/7 > 1/n = 1/7$. In general, we can construct schemes of arbitrary rate.

We also remark on an interesting analog between communication efficient secret sharing and the well-studied subject of regenerating codes \cite{Dimakis:2011iq, Tamo:2013hu, Rashmi:ii}. Consider a regenerating code of length $n$ that is able to correct  $r>1$ erasures. If only one erasure occurs, then compared to repairing from a minimum set of $n-r$ nodes, repairing from all the $n-1$ available nodes will significantly reduce the total amount of communication that occurs during the repair. In this sense, for both regenerating codes and communication efficient secret sharing, a key idea is to involve more available nodes/parties than the minimum required set during repair/decoding, for the purpose of reducing the repair/decoding bandwidth.

\subsection{Results}
In Section \ref{sec:lbnd}, we prove a tight information-theoretic lower bound on the decoding bandwidth, given a set of available parties $I \subset \{1, ..., n\}$. The bound implies that the decoding bandwidth decreases as $|I|$ increases. The lower bound applies to both perfect and ramp schemes and generalizes the lower bound in \cite{Wang08}. Particularly, we show that the overhead in communication for the case of $|I| = n$ is only a fraction $(n -r-z)/(n-z)$ of the communication overhead when $|I| = n-r$. 

In Section \ref{sec:shar}, we construct efficient secret sharing schemes using the ideas described in Section \ref{sec:mex}. Our construction utilizes Shamir's scheme and achieves the optimal decoding bandwidth universally for all $I \in \mathcal{A}$. Additionally, the construction preserves the simplicity of Shamir's scheme and is efficient in terms of both space and computation. Specifically, the scheme achieves optimal space efficiency, and requires the same field size as Shamir's scheme. Encoding and decoding the scheme is also similar to encoding and decoding Shamir's scheme. The scheme shows that our lower bound in Section  \ref{sec:lbnd} is uniformly tight. Interestingly, the scheme also generalizes the construction in a recent independent work \cite{Bitar15}. However, the flexibility of our framework allows improved efficiency in terms of computation, decoding delay and partial decoding. 


In Section \ref{sec:rs}, we construct another secret sharing scheme from Reed-Solomon codes. The scheme achieves the optimal decoding bandwidth  when $|I|=n$ and $|I|=n-r$. The decoder of the scheme has a simpler structure compared to the decoder of the previous scheme, and therefore is advantageous in terms of implementation. The scheme also offers a stronger level of reliability in that it allows decoding even if more than $r$ shares are partially lost. In Section \ref{sec:rdm} we present a scheme from random linear codes that achieves the optimal decoding bandwidth universally.

Finally, in the application of storage where each party is regarded as a disk, it is desirable to optimize the efficiency of disk operations. Our lower bound on the decoding bandwidth is naturally a lower bound on the number of symbol-reads from disks during decoding. In all of our schemes, the number of symbol-reads during decoding equals to the amount of communication. Therefore, our schemes are also optimal in terms of disk operations. In addition, by involving more than the minimum number of disks for decoding, our schemes balance the load at the disks and achieve a higher degree of parallelization.

\section{Secret Sharing Schemes}


Consider the problem of storing a secret message $\bm{m}$ securely and reliably into $n$ shares, so that 1) $\bm{m}$ can be recovered from any $n-r$ shares, and 2) any $z$ shares do not reveal any information about $\bm{m}$, i.e., they are statistically independent. Such a scheme is called a threshold secret sharing scheme, defined formally as follows. Let $\mathcal{Q}$ be a general $Q$-ary alphabet, i.e., $|\mathcal{Q}|=Q$. Denote by $[n] = \{1, ..., n\}$. For any index set $I \subset [n]$ and a vector $\bm{c} = (c_1, ..., c_n)$, denote by $\bm{c}_I = (c_i)_{ i \in I}$.
\begin{definition}\label{def:ss}
An $(n,k,r,z)_\mathcal{Q}$ secret sharing scheme consists of a randomized encoding function $F$ that maps a secret $\bm{m}  \in \mathcal{Q}^k$ to $\bm{c} = (c_1,...,c_n) = F(\bm{m}) \in \mathcal{Q}^n$, such that
\begin{itemize}
\item[1)] (Reliability) The secret $\bm{m}$ can be decoded from any $n-r$ shares (entries) of $\bm{c}$. This guarantees that $\bm{m}$ is recoverable in the loss of any $r$ shares. Formally,
\begin{align}\label{def:re}
H(\bm{m}|\bm{c}_I) = 0, \ \ \ \forall I \subset [n], \ |I| = n-r.
\end{align}
Therefore for any $I \subset [n]$, $|I| = n-r$, there exists a decoding function $D^*_I: \mathcal{Q}^{n-r} \to \mathcal{Q}^k$ such that $D^*_I(\bm{c}_I) = \bm{m}$.
\item[2)] (Secrecy) Any $z$ shares of $\bm{c}$ do not reveal any information about $\bm{m}$. This guarantees that $\bm{m}$ is secure if any $z$ shares are exposed to an eavesdropper. Formally,
\begin{align}\label{def:sr}
H(\bm{m}|\bm{c}_I) = H(\bm{m}), \ \ \ \forall I \subset [n], \  |I| = z.
\end{align}
\end{itemize}
\end{definition}

Define the \emph{rate} of a scheme to be $k/n$, which measures the space efficiency. The following proposition gives an upper bound on the rate. 
\begin{proposition}
For any $(n,k,r,z)_\mathcal{Q}$ secret sharing scheme, it follows that
\begin{align}\label{prop:cap}
k \le n - r - z,
\end{align}
and so the rate of the scheme is at most $\frac{n-r-z}{n}$.
\end{proposition}
\begin{proof}
Let the message $\bm{m}$ be uniformly distributed, then
\begin{align}
k  = H(\bm{m})  &= H(\bm{m}|\bm{c}_{[z]}) \label{eq:po1} \\& \le H(\bm{m},\bm{c}_{[n-r]}|\bm{c}_{[z]})\nonumber\\
& = H(\bm{m}|\bm{c}_{[n-r]}, \bm{c}_{[z]}) + H(\bm{c}_{[n-r]}|\bm{c}_{[z]})\label{eq:po2}\\
& = H(\bm{c}_{[n-r]}|\bm{c}_{[z]}) \label{eq:po3} \\
& = H(\bm{c}_{\{z+1, ..., n-r\}}) \le n-r-z, \nonumber
\end{align}
where (\ref{eq:po1}) follows from the security requirement, (\ref{eq:po2}) follows from the chain rule, and (\ref{eq:po3}) follows from the reliability requirement.
\end{proof}

A secret sharing scheme is \emph{rate-optimal} if it achieves equality in (\ref{prop:cap}). Note that the scheme is a \emph{perfect scheme} if $z=n-r-1$ and is a \emph{ramp scheme} otherwise. 
Rate-optimal perfect secret sharing schemes are studied in the seminal work by Shamir \cite{Shamir:1979vo}, and are later generalized to ramp schemes \cite{Blakley84, Yamamoto86}. 
Note that by (\ref{prop:cap}) the rate of any perfect scheme is at most $1/n$ as $k=1$. Any scheme of a higher rate is necessarily a ramp scheme. 
\section{Lower Bound on Communication Overhead}\label{sec:lbnd}

Suppose that the $n$ shares of the secret are stored by $n$ parties or distributed storage nodes\footnote{In what follows we do not distinguish between parties and nodes.}, and a user wants to decode the secret. By Definition \ref{def:ss}, the user can connect to any $n-r$ nodes and download one share, i.e., one $Q$-ary symbol, from each node. Therefore, by communicating $n-r$ symbols, the user can decode a secret of $k \le n-r-z$ symbols. It is clear that a communication overhead of $z$ symbols occurs during \emph{decoding}. The question is, whether it is possible to reduce the communication overhead. We answer this question affirmatively in the remaining part of the paper. 

There are two key ideas for improving the communication overhead. Firstly, in many practical scenarios and particularly in distributed storage systems, often time more than $n-r$ nodes are available. In this case, it is not necessary to restrict the user to download from only $n-r$ nodes. Secondly, it is not necessary to download the complete share stored by the node. Instead, it may suffice to communicate only a part of the share or, in general, a function of the share. In other words, a node can preprocess its share before transmitting it to the user.

 Motivated by these ideas, for any $I \subset [n]$, $|I| \ge n-r$, define a class of \emph{preprocessing functions} $E_{I,i} : \mathcal{Q} \to \mathcal{S}_{I,i}$,  where $|\mathcal{S}_{I,i}| \le |\mathcal{Q}|$, that maps $c_i$ to $e_{I,i}=E_{I,i}(c_i)$. Let $\bm{e}_I = ( e_{I,i})_{ i \in I }$, and define a class of \emph{decoding functions} $D_{I} : \prod_{i \in I} \mathcal{S}_{I,i} \to \mathcal{Q}^k$, such that $D_{I} (\bm{e}_I) = \bm{m}$. For a naive example, consider any $I$ such that $|I| = n-r$. Then for $i \in I$, we can let $\mathcal{S}_{I,i} = \mathcal{Q}$, let $E_{I,i}$ be the identity function, and let $D_I$ be the naive decoding function $D^*_{I}$ described in Definition \ref{def:ss}. In the remaining paper, when $I$ is clear from the context, we will suppress it in the subscripts of $\mathcal{S}_{I,i}, E_{I,i}$, $e_{I,i}$ and $\bm{e}_I$, and denote them by $\mathcal{S}_i$, $E_i$, $e_i$ and $\bm{e}$ instead. We now formally define the notion of communication overhead in decoding. Note that all $\log$ functions in the paper are base $Q$.

\begin{definition}
For any $I$ such that $|I| \ge n-r$, define the communication overhead function to be $\emph{\text{CO}}(I) = \sum_{i \in I} \log |\mathcal{S}_{I,i}| - k$. Namely, $\emph{CO}(I)$ is the amount of extra information, measured in $Q$-ary symbols, that one needs to communicate in order to decode a secret of $k$ symbols, provided that the set of available shares is indexed by $I$.
\end{definition}
The following result provides a lower bound on the communication overhead function. It generalizes the lower bound in \cite{Wang08} for perfect schemes, i.e., schemes with $k=1$. 
\begin{theorem}\label{th:upbnd}
For any $(n,k,r,z)_{\mathcal{Q}}$ secret sharing scheme with preprocessing functions $\{E_{I,i}\}_{i \in [n], |I| \ge n-r}$ and decoding functions $\{D_I\}_{|I| \ge n-r}$, it follows that 
\begin{align}
\emph{CO}(I) \ge \frac{k z}{|I|-z}.\label{eq:upbnd}
\end{align}
\end{theorem}
\begin{proof}
Consider arbitrary $I = \{i_1, ..., i_{|I|}\}$ such that $|I| \ge n-r$. 
Assume without loss of generality that $|\mathcal{S}_{i_1}| \le |\mathcal{S}_{i_2}| \le ... \le |\mathcal{S}_{i_{|I|}}|$. Recall that $\bm{e}_I = (e_{i_1}, ..., e_{i_{|I|}})$ is the output of the preprocessing functions.
\begin{align}
H(e_{i_1}, ..., e_{i_{|I|-z}}) & \stackrel{(a)}{\ge}  H(e_{i_1}, ..., e_{i_{|I|-z}} | e_{i_{|I|-z+1}}, ..., e_{i_{|I|}} ) \nonumber\\
& \stackrel{(b)}{=} H(e_{i_1}, ..., e_{i_{|I|-z}} | e_{i_{|I|-z+1}}, ..., e_{i_{|I|}} ) + H(\bm{m} | e_{i_1}, ...,  e_{i_{|I|}} )\nonumber\\
& \stackrel{(c)}{=} H( \bm{m} , e_{i_1}, ..., e_{i_{|I|-z}}| e_{i_{|I|-z+1}}, ..., e_{i_{|I|}} )\nonumber\\
& \stackrel{}{\ge} H(\bm{m}|e_{i_{|I|-z+1}}, ..., e_{i_{|I|}})\nonumber\\
& \stackrel{(d)}{=} H(\bm{m}) = k, \label{eq:entk}
\end{align}
where (a) follows from conditioning reduces entropy, (b) follows from (\ref{def:re}), (c) follows form the chain rule, and (d) follows from (\ref{def:sr}). Therefore it follows from (\ref{eq:entk}) that 
\begin{align*}
\prod_{j=1}^{|I|-z} |\mathcal{S}_{i_j}| \ge Q^{H(e_{i_1}, ..., e_{i_{|I|-z}})} \ge Q^k,
\end{align*}
and so 
\begin{align}\label{eq:o1}
\sum_{j=1}^{|I|-z} \log|\mathcal{S}_{i_j}| \ge k.
\end{align}
It then follows from $|\mathcal{S}_{i_1}| \le ... \le |\mathcal{S}_{i_{|I|}}|$ that,
\begin{align*}
\log|\mathcal{S}_{i_{|I|-z}}| \ge  \frac{k}{|I|-z},
\end{align*} 
and that,
\begin{align}
\log|\mathcal{S}_{i_{|I|-z+j}}| \ge \log|\mathcal{S}_{i_{|I|-z}}|  \ge  \frac{k}{|I|-z}, \ \ \ \ \  j=1,...,z \label{eq:o2}.
\end{align}
Combining (\ref{eq:o1}) and (\ref{eq:o2}) we have,
\begin{align*}
\text{CO}(I) = \sum_{j=1}^{|I|} \log|\mathcal{S}_{i_j}| -k \ge \frac{kz}{|I|-z}.
\end{align*}
\end{proof}

The \emph{decoding bandwidth} is defined to be the total amount of $Q$-ary symbols the user downloads from the nodes, which equals $CO(I)+k$. Theorem \ref{th:upbnd} suggests that the communication overhead and the decoding bandwidth decrease as the number of available nodes increases.

For rate-optimal schemes, Theorem \ref{th:upbnd} implies that if $|I| = n-r$, then the communication overhead is at least $z$, i.e., the user needs to download the complete share from each available node. The naive decoding function $D^*_I$ in Definition \ref{def:ss} trivially achieves this bound. The more interesting scenario is the regime that $|I| > n-r$. In this case, if (\ref{eq:upbnd}) is tight, then one can achieve a non-trivial improvement on decoding bandwidth compared to the naive decoder $D^*_I$. When $k=1$ (i.e., for perfect schemes) and fixing any $d > n-r$, \cite{Wang08} constructs a rate-optimal  scheme that achieves the lower bound (\ref{eq:upbnd}) for any $I$ such that $|I| = d$. However, several interesting and important questions remain open. Firstly, is the lower bound uniformly tight, or in other words, is it possible to construct a scheme that achieves (\ref{eq:upbnd}) universally for any $I$ such that $|I|\ge n-r$ (note that the scheme in \cite{Wang08} does not achieve the lower bound when $|I| \ne d$)? Secondly, is the bound tight when $k>1$ (i.e., for ramp schemes) and how to design such schemes? We answer these questions in the following section.

\section{Construction from Shamir's scheme}\label{sec:shar}
In this section we construct a rate-optimal scheme that achieves the optimal decoding bandwidth universally for all possible $I$, i.e., all sets of available nodes. This implies that the lower bound in Theorem \ref{th:upbnd} is uniformly tight. The scheme is based on Shamir's scheme  and preserves its simplicity and efficiency. The scheme is flexible in the parameters $n$, $k$, $r$, $z$ and hence is flexible in rate. 

We first refer the readers to Figure \ref{fig:ours} for an example of the scheme, and use it to describe the general idea of the construction. To construct a scheme that achieves the optimal decoding bandwidth when $d$ nodes are available, for all $d \in \mathcal{D}$, we design a set of polynomials of different degrees. Particularly,  for all $d \in \mathcal{D}$, we design a number of polynomials of degree exactly $d-1$, and store one evaluation of each polynomial at each node. For each polynomial, exactly $z$ of its coefficients are independent keys in order to meet the secrecy requirement. The remaining coefficients encode ``information'': for the highest-degree (e.g., degree $d_{\max}-1$, where $d_{\max}=\max_{d \in \mathcal{D}} d$) polynomials, their coefficients encode the entire message; for other polynomials, say $g(x)$, the information encoded in the coefficients of $g(x)$ is the high-degree coefficients of the polynomials of degree higher than $g(x)$. Such an arrangement of the coefficients enables decoding in a successive manner. Consider decoding when $d$ nodes are available, implying that $d$ evaluations of each polynomial are known and hence all polynomials of degree $d-1$ can be interpolated. Then, roughly speaking, the arrangement ensures that the high-degree coefficients of some higher-degree polynomials are known, so that the remaining unknown parts of these polynomials can be interpolated. This in turn allows to decode coefficients for additional high-degree polynomials and thus to interpolate them. The chain continues until all polynomials of degree higher than $d-1$ are interpolated, implying that the message is decoded. Note that no polynomials of degree smaller than $d-1$ are interpolated, and therefore the keys associated with them are not decoded. This leads to the saving in decoding bandwidth and in fact this amount is the best one can expect to save, so that the scheme achieves the optimal bandwidth. Below we describe the scheme formally.

\subsection{Encoding}\label{sec:sharenc}
Consider arbitrary parameters $n,r,z$, $\mathcal{D}$ and let $k=n-r-z$. We assume that $n-r \in \mathcal{D}$ since it is implied by the reliability requirement. Choose any prime power $q>n$, the scheme is $\mathbb{F}_q$-linear over share alphabet $\mathcal{Q}=\mathbb{F}_q^b$, where $b$ is the number of ($\mathbb{F}_q$) symbols stored by each node. The message $\bm{m}$ is a vector over $\mathbb{F}_q$ of length $|\bm{m}|=kb$.  The choice of $b$ is determined by $\mathcal{D}$ in the following way. Let $|\bm{m}|$ be the least common multiple of $\{d-z: d \in \mathcal{D}\}$, i.e., the smallest positive integer that is divisible by all elements of the set. Note that indeed $|\bm{m}|$ is a multiple of $k=n-r-z$, and we let $b = \frac{|\bm{m}|}{k}$. This is the smallest choice of $|\bm{m}|$ (and thus $b$) that ensures when $d \in \mathcal{D}$ nodes are available, that the optimal bandwidth, measured by the number of $\mathbb{F}_q$ symbols, is an integer.


We now construct $b$ polynomials over $\mathbb{F}_q$, evaluate each of them at $n$ non-zero points, and let every node stores an evaluation of each polynomial. Let $\mathcal{D} = \{d_1, d_2, ..., d_{|\mathcal{D}|} \}$, such that $n \ge d_1 > d_2 > ... > d_{|\mathcal{D}|}=n-r$. For $i \in |\mathcal{D}|$, let
\begin{align}
p_i = \left\{ \begin{array}{ll}
\frac{|\bm{m}|}{d_1-z} & i=1\\
\frac{|\bm{m}|}{d_i-z} - \frac{|\bm{m}|}{d_{i-1}-z} & i>1
\end{array}\right.
\end{align}
We construct $p_i$ polynomials of degree $d_i-1$. For all polynomials, their $z$ lowest-degree coefficients are independent random keys. We next define the remaining $d_i-z$ non-key coefficients. We first define them for the highest degree polynomials, and then recursively define them for the lower degree polynomials. For $i=1$, the non-key coefficients of the polynomials of degree $d_i-1$ are message symbols. Note that there are $|\bm{m}|$ message symbols and $\frac{|\bm{m}|}{d_1-z}$ polynomials of degree $d_1-1$. Each such polynomial has $d_1-z$ non-key coefficients and so there are exactly enough coefficients to encode the message symbols. For $i>1$, the non-key coefficients encode the degree $d_i$ to $d_{i-1}-1$ coefficients of all higher (than $d_{i}-1$) degree polynomials. Note that there are $\sum_{j=1}^{i-1} p_j = \frac{|\bm{m}|}{d_{i-1}-z}$ higher degree polynomials and so the total number of coefficients to encode is $(d_{i-1}-d_{i}) \frac{|\bm{m}|}{d_{i-1}-z}$. On the other hand, there are $p_i$ polynomials of degree $d_i-1$, each of them has $d_i-z$ non-key coefficients, and so the total number of non-key coefficients is $(d_i-z)\left(\frac{|\bm{m}|}{d_i-z} - \frac{|\bm{m}|}{d_{i-1}-z}\right)$. It is trivial to verify that the two numbers are equal and so there is exactly enough coefficients to encode. Note that the specific way to map the coefficients is not important and any 1-1 mapping suffices. Finally, evaluate each polynomial at $n$ non-zero points and store an evaluation of each polynomial at each node. This completes the scheme. Note that indeed the total number of polynomials is $\sum_{i=1}^{|\mathcal{D}|} p_i = \frac{|\bm{m}|}{d_{|\mathcal{D}|}-z} = \frac{|\bm{m}|}{k} =b$, implying that the scheme is rate-optimal.

\subsection{Decoding}
 For any $d_i \in \mathcal{D}$, we describe the decoding algorithm of the scheme when $d_i$ nodes are available. It achieves the optimal decoding bandwidth, and since $d_{|\mathcal{D}|}=n-r$ it implies that the scheme meets the reliability requirement. We first interpolate all polynomials of degree $d_i-1$. After that for all polynomials of degree $d_{i-1}-1$, their coefficients of degree larger than $d_i-1$ are known (as they are encoded in the coefficients of the polynomials of degree $d_i-1$) and so they can be interpolated. In general, for $j \le i$, once the polynomials of degree between $d_j-1$ and $d_i -1$ are interpolated, then for the polynomials of degree $d_{j-1}-1$, their coefficients of degree larger than $d_i-1$ are known by construction and so they can be interpolated. Therefore we can successively interpolate the polynomials of higher degree until the polynomials of degree $d_1-1$ are interpolated and so the message symbols are decoded. The total number of $\mathbb{F}_q$ symbols communicated is $d_i \sum_{j=1}^{i} p_j = d_i \frac{|\bm{m}|}{d_i-z}$. By Theorem \ref{th:upbnd}, the decoding bandwidth is at least $$ |\bm{m}| + \frac{kbz}{d_i-z}= kb + \frac{kbz}{d_i-z} = kb\left(1+\frac{z}{d_i-z}\right) = \frac{d_i |\bm{m}|}{d_i-z}$$ $\mathbb{F}_q$ symbols. Therefore the optimal bandwidth is achieved. 

\subsection{Secrecy}
We show that the scheme is secure against $z$ eavesdropping nodes. Since each polynomial individually is a Shamir's scheme, the secrecy of the scheme derives from the secrecy of Shamir's scheme. The main idea is to show that if these polynomials are combined, the resulting scheme is still secure. We first prove a simple lemma.
\begin{lemma}\label{lem:info}
Consider random variables $M_1$, $M_2$, $K_1$, $K_2$ such that $K_2$ is independent of $\{M_1,K_1\}$. For $i=1,2$ Let $F_i$ be a deterministic function of $M_i, K_i$. If $I(M_1;F_1)=0$ and $I(M_2;F_2)=0$, then $I(M_1; F_1, F_2)=0$. In addition, if $K_1$ is independent of $M_2$, then $I(M_1,M_2;F_1,F_2)=0$.
\end{lemma}
\begin{proof} We start with the first statement. Since $F_2$ is a function of $K_2, M_2$ but $K_2$ is independent of $\{M_1, K_1, F_1\}$, it follows that 
 $F_2$ is independent of $\{M_1, K_1, F_1\}$ conditioning on $M_2$, implying the Markov chain $\{M_1, K_1, F_1\} \to M_2 \to F_2$. Therefore, $I(M_1, K_1, F_1, M_2; F_2) = I(M_2; F_2) = 0$, i.e., $F_2$ and $\{M_1, K_1, F_1, M_2\}$ are independent. Hence $I(M_1; F_1, F_2) = I(M_1; F_2) + I(M_1; F_1|F_2) \stackrel{(a)}{=} I(M_1; F_1|F_2) \stackrel{(b)}{=} I(M_1; F_1) =0 $, where $(a)$ and $(b)$ follows from the fact that $F_2$ is independent from $\{M_1, F_1\}$.
 
 To prove the second statement, note that since $K_1$ is independent of $M_2$ and that $F_1$ is a function of $M_1, K_1$, we have the Markov Chain $M_2 \to M_1 \to F_1$, by which it follows that $I(M_1, M_2; F_1) = I(M_1; F_1) =0$.  Similarly because $K_2$ is independent of $\{M_1,K_1,F_1\}$ and that $F_2$ is a function of $M_2, K_2$, we have the Markov Chain $\{M_1,F_1\} \to M_2 \to F_2$. By this chain it follows that $I(M_1, F_1, M_2;F_2) = I(M_2;F_2)=0$, i.e., $\{M_1,F_1,M_2\}$ is independent of $F_2$. Therefore $I(M_1,M_2; F_2| F_1) = 0$ and so $I(M_1,M_2;F_1,F_2) = I(M_1,M_2;F_1) + I(M_1,M_2; F_2| F_1) =0$.
\end{proof}
Suppose that the adversary compromises $z$ nodes and obtains $z$ evaluations of each polynomial. Consider the $i$-th polynomial in the order that we define them, let $\bm{f}_i$ denote the adversary's observation of this polynomial, let $\bm{k}_i$ denote the key coefficients of this polynomial and let $\bm{m}_i$ denote the non-key coefficients.
The secrecy of Shamir's scheme implies that
\begin{align}
I(\bm{m}_i; \bm{f}_i)=0, \ \ \  i =1, ..., b.
\end{align}

Consider the first $p_1$ polynomials which are polynomials of the highest degree $d_1-1$. By construction, $\bm{m}_1,...,\bm{m}_{p_1}$ exactly encode the message $\bm{m}$.  We invoke  Lemma \ref{lem:info} by regarding $\bm{m}_1, \bm{k}_1$, $\bm{f}_1, \bm{m}_2, \bm{k}_2$ and $\bm{f}_2$ as $M_1, K_1, F_1, M_2, K_2$ and $F_2$.  By the second statement of the lemma it follows that $I(\bm{m}_1, \bm{m}_2; \bm{f}_1, \bm{f}_2)=0$. Inductively, for $1<i< p_1$, suppose that $I(\bm{m}_1, ..., \bm{m}_i; \bm{f}_1, ..., \bm{f}_i)=0$. We regard $\{\bm{m}_1,...,\bm{m}_i\}$ as $M_1$, $\{ \bm{k}_1,...,\bm{k}_i \}$ as $K_1$, $\{ \bm{f}_1,...,\bm{f}_i \}$ as $F_1$, and regard $\bm{m}_{i+1}, \bm{k}_{i+1} ,\bm{f}_{i+1}$ as $M_2, K_2, F_2$. It follows from Lemma~\ref{lem:info} that $I(\bm{m}_1, ..., \bm{m}_{i+1}; \bm{f}_1, ..., \bm{f}_{i+1})=0$. By induction we have $I(\bm{m}_1, ..., \bm{m}_{p_1}; \bm{f}_1, ..., \bm{f}_{p_1})=0$.

We then regard $ \{\bm{m}_1, ..., \bm{m}_{p_1}\} \triangleq \bm{m}$ as $M_1$, $\{ \bm{k}_1,...,\bm{k}_{p_1} \}$ as $K_1$, $\{ \bm{f}_1,...,\bm{f}_{p_1} \}$ as $F_1$, and regard $\bm{m}_{p_1+1}$, $\bm{k}_{p_1+1}$, $\bm{f}_{p_1+1}$ as $M_2, K_2, F_2$. Then it follows from the first statement of Lemma \ref{lem:info} that $I(\bm{m}; \bm{f}_1, ..., \bm{f}_{p_1+1})=0$. Inductively, for $p_1 < i < b$, suppose that $I(\bm{m}; \bm{f}_1, ..., \bm{f}_{i})=0$. We regard $\bm{m}$ as $M_1$, $\{\bm{k}_1, ..., \bm{k}_{i}\}$ as $K_1$, $\{\bm{f}_1, ..., \bm{f}_{i}\}$ as $F_1$, and regard $\bm{m}_{i+1}, \bm{k}_{i+1}, \bm{f}_{i+1}$ as $M_2, K_2, F_2$.
By Lemma \ref{lem:info} we have $I(\bm{m}; \bm{f}_1, ..., \bm{f}_{i+1})=0$. By induction it follows that $I(\bm{m}; \bm{f}_1, ..., \bm{f}_b)=0$, implying that the adversary learns no information about the message $\bm{m}$. This completes the proof and we have the following theorem.
\begin{theorem} Let $\mathcal{D} \subset \{n-r,n-r+1,...,n\}$,
the encoding scheme constructed in Section \ref{sec:sharenc} is a rate-optimal $(n,k,r,z)$ secret sharing scheme. The scheme achieves the optimal decoding bandwidth when $d$ nodes participate in decoding, universally for all $d \in \mathcal{D}$. \end{theorem}
\subsection{Discussion} 
We remark on some other important advantages and properties of our construction. Firstly, the scheme also achieves the \emph{optimal number of symbol-reads from disks} in decoding. To see this, notice that the lower bound (\ref{eq:upbnd}) on communication overhead is also a lower bound on the number of $Q$-ary symbols that need to be read from disks during decoding.  The number of symbol-reads in the proposed scheme equals to the amount of communication. Therefore our scheme achieves the lower bound and hence is optimal.
 Secondly, compared to most existing schemes which decode from the minimum number of $n-r-z$ nodes, our scheme allows all available nodes (or more flexibly, any $d \in \mathcal{D}$ nodes) to participate in decoding and hence can help balance the load at the disks and achieves a higher degree of parallelization. Thirdly, the encoding and decoding of the scheme are similar to that of Shamir's scheme and therefore are efficient and practical. Particularly, the scheme works over the same field as Shamir's scheme.  Fourthly, the preprocessing functions only rely on $d=|I|$ instead of $I$, further simplifying implementation. Finally, the construction is flexible in the parameters, i.e., it works for arbitrary values of $n,r$ and $z$ and $\mathcal{D}$.

An important idea in our scheme is to construct polynomials of different degrees in order to facilitate decoding when different number of nodes are available. Similar ideas also appear in the schemes in \cite{Wang08,Zhang2012}. The main technique that enables the improvement of our schemes is a more careful and flexible design of the numbers and degrees of the polynomials, as well as the arrangement of their coefficients.

Our scheme maps the high-degree coefficients of the higher degree polynomials into the coefficients of the lower degree polynomials, whereas the specific mapping is not important and any 1-1 mapping suffices. In practice, the flexibility in choosing the specific mapping is helpful. Particularly, it is possible to improve the (computational) encoding complexity of the scheme substantially by choosing a mapping that maintains the order of the coefficients. Refer to Figure \ref{fig:ours} for an example. We need to compute $m_4x + m_5x^2 + m_6x^3$ in evaluating $g(x)$, and we can reuse this computation in evaluating $f(x)$, because $f(x)$ contains the same run of consecutive coefficients $m_4x^4 + m_5x^5 + m_6x^6$. This for example will save 2 multiplications and 2 additions. 

We also note that for all polynomials in our scheme, the $z$ lowest degree coefficients are independent keys. However, in general this is not necessary: in any polynomial, we can choose any consecutive $z$ coefficients to be independent keys, and use the remaining coefficients to encode information (i.e., message symbols and coefficients of higher degree polynomials). The resulting scheme is a still valid and achieves the optimal decoding bandwidth universally. Under this observation, we note that our scheme generalizes  the scheme in a recent independent work \cite{Bitar15}. Particularly, our scheme is equivalent to the scheme in \cite{Bitar15} if we require a specific coefficient mapping and let the $z$ highest (instead of lowest) coefficients of all polynomial to be keys\footnote{The scheme in \cite{Bitar15} also lets a node evaluate all polynomials at the same point, whereas this is not necessary in our framework.}.

As noted above, the flexibility of our scheme in choosing the coefficient mapping is beneficial in practice. Furthermore, we remark that choosing the lowest degree coefficients to be keys has several practical advantages: decoding the scheme involves sequentially interpolating the polynomials through multiple iterations, which can lead to undesirable delay especially when $|\mathcal{D}|$ is large. To mitigate this issue, we wish to decode the message symbols ``on the fly'' in each iteration. Specifically, if $d$ nodes are available, then each time a polynomial is interpolated, exactly $d$ new message and/or key symbols are decoded. Since the number of symbols decoded in each interpolation, the total number of message symbols and the total number of key symbols to be decoded are all fixed, there is a trade-off between the decoding order of the key and message symbols. 
The \emph{optimal trade-off} is to delay decoding the keys as much as possible, so that the maximum number of message symbols are decoded on the fly. Specifically, notice that by the time that a number of $i$ polynomials are interpolated, at least $zi$ key symbols are decoded since each polynomial introduces $z$ independent key coefficients for secrecy. The optimal trade-off is achieved if indeed exactly $zi$ keys are decoded, implying that $(d-z)i$ message symbols are decoded. Our scheme achieves this optimal trade-off by choosing the $z$ lowest degree coefficients to be keys. This is because by construction, only coefficients of degree higher than $d_{|\mathcal{D}|} =n-r>z$ will be mapped to the coefficients of the lower degree polynomials. Hence the key coefficients are never mapped, implying that the remaining information coefficients encode only message symbols. Therefore, at any moment during the decoding process, our scheme always decodes the maximum number of message symbols. In other words the decoding delay, measured in the number of iterations, averaged over all message symbols, is minimized. Moreover, the fact that each polynomial interpolation decodes a fixed number of $d-z$ new message symbols is helpful for implementation. On the other hand, note that choosing the $z$ highest degree coefficients to be keys implies that the keys will be mapped to the coefficients of lower degree polynomials. Hence the keys will be decoded earlier than necessary (since lower degree polynomials are interpolated earlier) and it is not possible to achieve the optimal trade-off. Consider the example in Figure \ref{fig:ours}, if we switch the keys to high degree coefficients, then the polynomials are  $f(x)=m_1+m_2x+m_3 x^2 + m_4 x^3 + m_5 x^4  + m_6 x^5 + k_1 x^6$, $g(x)=m_5 + m_6 x + k_1 x^2 + k_2 x^3$ and $h(x) = m_4 + k_2 x + k_3 x^2$. In the case that $d=4$ nodes are available, only 2 message symbols $m_5,m_6$ are decoded in the first iteration and the remaining 4 message symbols are decoded in the second (last) iteration. In comparison, the original scheme performs better by decoding 3 message symbols in each iteration. Finally, we remark that decoding the maximum number of message symbols on the fly is also beneficial in terms of partial decoding, i.e., decoding a subset of message symbols. In this case decoding can finish early if all symbols of interest are decoded, and our scheme will maximize the chance of finishing early.



\section{Construction from Reed-Solomon Codes}\label{sec:rs}
In this section we present another rate-optimal secret sharing scheme that achieves the optimal decoding bandwidth when all $n$ nodes are available. The scheme is flexible in the parameters and hence is flexible in rate.  The scheme is directly related to Reed-Solomon codes. Particularly, the encoding matrix of the scheme is a generator matrix of Reed-Solomon codes, and so the scheme can be decoded as  Reed-Solomon codes. This is an advantage over the scheme in the previous section, which requires recursive decoding. The scheme also provides a stronger level of reliability in the sense that it allows decoding even if more than $r$ shares are partially erased. On the other hand, unlike the previous scheme, this scheme does not achieve the optimal decoding bandwidth universally, but rather only for $d=n-r$ and $d=n$. However, we remark that the case that the $n$ nodes are available is particularly important because it correspond to the  best case in terms of decoding bandwidth and is arguably the most relevant case for the application of distributed storage, where the storage nodes are usually highly available.

\subsection{Encoding} \label{sec:enc}
Fix $k=n-r-z$, let $q> n(k+r)$ be a prime power, and let the share alphabet be $\mathcal{Q} = \mathbb{F}_q^{k+r}$.  
Note that each share is a length $k+r$ vector over $\mathbb{F}_q$. For $j=1,...,n$, denote the $j$-th share by $c_j = (c_{1,j}, ..., c_{k+r,j})$, where $c_{i,j} \in \mathbb{F}_q$. The secret message $\bm{m}$ is $k$ symbols over $\mathcal{Q}$ and therefore can be regarded as a length-$k(k+r)$ vector over $\mathbb{F}_q$, denoted by $(m_1,...,m_{k(k+r)})$. The encoder generates keys $\bm{k}=(k_1, ..., k_{kz}) \in \mathbb{F}_q^{kz}$ and $\bm{k}'=(k'_1,...,k'_{rz}) \in \mathbb{F}_q^{rz}$ independently and uniformly at random. The encoding scheme is linear over $\mathbb{F}_q$, and is described by an encoding matrix $G$ over $\mathbb{F}_q$:
\begin{align}\label{eq:enc}
 (c_{1,1}, ..., c_{1,n}, ..., c_{k+r,1} , ..., c_{k+r,n})
& = (m_1,...,m_{k(k+r)}, k_1,..., k_{kz},k'_1,...,k'_{rz}) G.
\end{align}

Note that $G$ has $k(k+r)+kz+rz=nk+rz$ rows and has $n(k+r)$ columns. 
In the following we discuss the construction of $G$ based on a Vandermonde matrix. 
We start with some notation. Let $\alpha_1, ..., \alpha_{n(k+r)}$ be distinct non-zero elements of $\mathbb{F}_q$, and let $v_{ij}=\alpha_j^{i-1}$, $i=1,...,nk+rz$, $j=1,...,n(k+r)$, then $V = (v_{ij})$ is a Vandermonde matrix of the same size as $G$. Suppose $\bm{f}=(f_0, ..., f_i)$ is an arbitrary vector with entries in $\mathbb{F}_q$, we denote by $\bm{f}[x]$ the polynomial $f_0 + f_1 x + ... + f_i x^i$ over $\mathbb{F}_q$ with indeterminate $x$. We construct a set of polynomials as follows:
\begin{align}
\bm{f}_i[x] & = x^{i-1} \hspace{30mm} i=1,...,kn,\label{eq:ele}\\
\bm{f}_{kn+i}[x] & = x^{i-1}\prod_{j=1}^{kn}(x-\alpha_j)\hspace{10mm}i=1,...,rz. \label{eq:flower}
\end{align}
Let $\bm{f}_i, i=1,...,kn+rz$ be the length-($kn+rz$) vectors over $\mathbb{F}_q$ corresponding to the polynomials. Stack the $\bm{f}_i$'s to obtain a sqaure matrix of size $(kn+rz)$:
\begin{align*}
T = \left(
\begin{array}{c}
\bm{f}_1\\
\vdots\\
\bm{f}_{kn+rz}
\end{array}
 \right)
\end{align*}
Finally, we complete the construction by setting
\begin{align*}
G=TV.
\end{align*}

\begin{example}\label{ex:1}
Consider the setting that $n=3, r=1, z=1$ and $k=n-r-z=1$. Let $q=7$ and $\mathcal{Q} = \mathbb{F}_q^2$. Then $\bm{m} = (m_1, m_2)$, $\bm{k} = (k_1)$ and $\bm{k}' = (k'_1)$. Construct a Vandermonde matrix over $\mathbb{F}_q$ as
\begin{align}
V = \left( 
\begin{array}{cccccc}
1 & 1 & 1 & 1 & 1 & 1\\
1 & 2 & 3 & 4 & 5 & 6\\
1 & 4 & 2 & 2 & 4 & 1\\
1 & 1 & 6 & 1 & 6 & 6
\end{array}
 \right).
\end{align}
Construct polynomials $\bm{f}_1[x] = 1$, $\bm{f}_2 [x] = x$, $\bm{f}_3 [x] =x^2$ and 
\begin{align*}
\bm{f}_4 [x] = (x-1)(x-2)(x-3) = 1 + 4x + x^2 + x^3.
\end{align*}
Therefore,
\begin{align*}
T = \left(
\begin{array}{c}
\bm{f}_1 \\ \bm{f}_2 \\ \bm{f}_3 \\ \bm{f}_4
\end{array}
\right) = 
\left(
\begin{array}{cccc}
1 & 0 & 0 & 0\\
0 & 1 & 0 & 0\\
0 & 0 & 1 & 0\\
1 & 4 & 1 & 1
\end{array}
\right),
\end{align*}
and the encoding matrix is given by
\begin{align*}
G = TV = \left(
\begin{array}{cccccc}
1 & 1 & 1 & 1 & 1 & 1\\
1 & 2 & 3 & 4 & 5 & 6\\
1 & 4 & 2 & 2 & 4 & 1\\
0 & 0 & 0 & 6 & 3 & 4
\end{array}
\right).
\end{align*}
\end{example}

The properties of $G$ are discussed in the following lemma.
\begin{lemma}\label{lemma:G}
Regard $G$ as a block matrix
\begin{align*}
G = \left( 
\begin{array}{cc}
G_{11} & G_{12}\\
G_{21} & G_{22}
\end{array}
 \right),
\end{align*}
where $G_{11}$ has size $kn \times kn$, $G_{12}$ has size $kn \times rn$, $G_{21}$ has size $rz\times kn$, and $G_{22}$ has size $rz \times rn$. Then,
\begin{itemize}
\item[(i)] Any $(n-r)(k+r)$ columns of $G$ are linearly independent.
\item[(ii)] $G_{11}$ is a Vandermonde matrix.
\item[(iii)] $G_{21}=0$.
\item[(iv)] Any $rz$ columns of $G_{22}$ are linearly independent.
\end{itemize}
\end{lemma}
\begin{proof}
By construction, the polynomials $\bm{f}_i[x], i=1,...,kn+rz$ have distinct degrees and therefore are linearly independent. Therefore the rows of $T$ are linearly independent and so $T$ is full rank. This implies that the row space of $G$ is the same as the row space of $V$. The row space of $V$ is a linear $(nk+nr, nk+rz)$ MDS code\footnote{In fact this is the Reed-Solomon code.} because that $V$ is a Vandermonde matrix. Note that $nk+rz = (n-r)(k+r)$, and so the row space of $G$ is a linear $(nk+nr, (n-r)(k+r) )$ MDS code. This proves (i).

To prove (ii), note that by (\ref{eq:ele}), the first $kn$ rows of $G$ are exactly the first $kn$ rows of $V$. Therefore $G_{11}$ is a Vandermonde matrix.

To prove (iii), note that by construction the $(i,j)$-th entry of $G_{21}$ equals $\bm{f}_{kn+i} [\alpha_j] $. By (\ref{eq:flower}), $\alpha_j$ is a root of $\bm{f}_{kn+i}[x]$, for $i=1, ..., rz$, $j=1, ..., kn$. Hence $G_{21}=0$.

Finally we prove (iv). By construction the $(i,j)$-th entry of $G_{22}$ equals
\begin{align}\label{eq:f*}
\bm{f}_{kn+i} [\alpha_{kn+j}] =  \alpha_{kn+j}^{i-1} \prod_{l=1}^{kn}(\alpha_{kn+j}-\alpha_l)= \alpha_{kn+j}^{i-1} \bm{f}^*[ \alpha_{kn+j}],
\end{align}
where $\bm{f}^*[x]=\prod_{l=1}^{kn}(x-\alpha_l)$. Since  $\alpha_1, ..., \alpha_{(k+r)n}$ are distinct elements, it follows that $\bm{f}^*[\alpha_{kn+j}] \ne 0$, for $j=1,...,rn$.
Let $1\le j_1 < j_2 < ... < j_{rz} \le rn$ and consider the submatrix formed by the $j_1$-th,...,$j_{rz}$-th columns of $G_{22}$. By (\ref{eq:f*}), the $l$-th column of the submatrix are formed by consecutive powers of $\alpha_{kn+j_l}$, scaled by $\bm{f}^*[\alpha_{kn+j_l}]$. Therefore the determinant of the submatrix is $\prod_{l=1}^{rz} \bm{f}^*[\alpha_{kn+j_l}] \prod_{1 \le u < v \le rz }(\alpha_{kn+j_v} - \alpha_{kn+j_u}) \ne 0$. This shows that any $rz$ columns of $G_{22}$ are linearly independent.
\end{proof}

\subsection{Decoding}\label{sec:dec}
We describe the decoding procedure for two cases: 1) $|I|=n$, i.e., all nodes  are available, and 2) $|I|<n$. 
First consider the case that $|I|=n$, i.e., $I=[n]$. In order to decode, for this case it suffices to read and communicate the first $k$ symbols over $\mathbb{F}_q$ from each share. Formally, the user downloads $\bm{e} = (c_{1,1}, ..., c_{1,n}, ..., c_{k,1}, ..., c_{k,n})$. By Lemma \ref{lemma:G}.(ii), $G_{11}$ is invertible. Denote the inverse of $G_{11}$ by $G^{-1}_{11}$, then the secret can be recovered by
\begin{align*}
 \bm{e} G^{-1}_{11} \stackrel{(e)}{=} (m_1, ..., m_{k(k+r)}, k_1, ..., k_{kz}),
\end{align*}
where (e) follows from (\ref{eq:enc}) and Lemma \ref{lemma:G}.(iii).  The decoding process involves communicating $kn$ symbols from $\mathbb{F}_q$. 
The communication overhead is $kz$ symbols over $\mathbb{F}_q$ or $\frac{kz}{k+r}=\frac{kz}{n-z}$ $\mathcal{Q}$-ary symbols, which achieves the lower bound (\ref{eq:upbnd}) and therefore is optimal.

Next consider the case that $ n-r \le |I|<n$. Select an arbitrary subset $I'$ of $I$ of size $n-r$, and download the complete share stored by the nodes in $I'$.  Hence, the downloaded information $\bm{e}$ is a length-$(n-r)(k+r)$ vector over $\mathbb{F}_q$. By Lemma \ref{lemma:G}.(i), it follows that any $(n-r)(k+r)$ columns in $G$ are linearly independent and therefore the submatrix formed by these columns is invertible.  The secret $\bm{m}$ can then be recovered by multiplying $\bm{e}$ with the inverse. An alternative way to decode the secret is to notice that $G$ is an encoding matrix of a $(nk+nr,nk+rz)$ Reed-Solomon code over $\mathbb{F}_q$. Therefore one may employ the standard decoder of Reed-Solomon code to correct any $r(k+r)$ erasures or $\lfloor r(k+r)/2 \rfloor$ errors of symbols over $\mathbb{F}_q$. Note that when at most $r$ nodes are unavailable , we regard their shares as erased and there are at most $r(k+r)$ erasures of symbols over $\mathbb{F}_q$, and therefore can be corrected. In general, any $r(k+r)$ erasures or $\lfloor r(k+r)/2 \rfloor$ errors are correctable even if they occur to more than $r$ nodes.
The decoding process involves communicating $nk+rz$ symbols of $\mathbb{F}_q$. The communication overhead  is $(n-r)(k+r)-k(k+r)=z(k+r)$ symbols over $\mathbb{F}_q$, or $z$ symbols over $\mathcal{Q}$, which achieves the lower bound (\ref{eq:upbnd}) if and only if $|I|=n-r$.
\subsection{Analysis}
\begin{theorem}\label{th:con1}
The encoding scheme constructed in Section \ref{sec:enc} is a rate-optimal $(n,k,r,z)$ secret sharing scheme. The scheme achieves the optimal decoding bandwidth when $d$ nodes participate in decoding, for $d =n$ or $d=n-r$.
\end{theorem}
\begin{proof}
We need to verify that the encoding scheme meets the reliability requirement and the security requirement of a secret sharing scheme, formally defined in Definition \ref{def:ss}. Explicit decoding scheme and its communication overhead are discussed in Section \ref{sec:dec} and therefore the reliability requirement is met.  The scheme is rate-optimal because $k=n-r-z$. We only need to show that the encoding scheme is secure. 
To this end, we first show that $H( \bm{k}, \bm{k}' |\bm{c}_I, \bm{m})=0$, for all $I$ such that $|I| = z$. In other words, the random symbols generated by the encoder are completely determined by $\bm{c}_I$ and the secret. Denote the submatrix formed by the first $k(k+r)$ rows of $G$ by $G_{\text{top}}$ and the submatrix formed by the remaining $(k+r)z$ rows of $G$ by $G_{\text{low}}$. Consider any $I = \{ i_1, ..., i_{z} \}$, and let $\bm{c}_I = (c_{1,i_1}, ..., c_{1,i_{z}}, ..., c_{k+r, i_1}, ..., c_{k+r, i_z})$. It then follows from (\ref{eq:enc}) that
\begin{align*}
\bm{c}_I = (m_1, ..., m_{k(r+k)}) G_{\text{top}, I} + (k_1, ..., k_{kz}, k'_1, ..., k'_{rz}) G_{\text{low},I},
\end{align*}
where $G_{\text{top}, I}$ is the submatrix formed by the subset of columns in $\{ i+j | i \in I, j=0,n,...,(k+r-1)n \}$ of $G_{\text{top}}$, and $G_{\text{low}, I}$ is the submatrix formed by the same subset of columns of $G_{\text{low}}$. Therefore, written concisely, 
\begin{align}
(\bm{k} \ \bm{k}') G_{\text{low},I} = \bm{c}_I - \bm{m} G_{\text{top}, I}. \label{eq:kk'rs}
\end{align}
To study the rank of $G_{\text{low},I}$, note that it is a square matrix of size $(k+r)z$, and we regard it as a block matrix
\begin{align}\label{eq:Glow}
G_{\text{low},I} = \left( 
\begin{array}{cc}
G'_{11} & G'_{12}\\
G'_{21} & G'_{22}
\end{array}
 \right),
\end{align}
where $G'_{11}$ has size $kz \times kz$, $G'_{12}$ has size $ kz \times rz$, $G'_{21}$ has size $rz \times kz$ and $G'_{22}$ has size $rz \times rz$. By Lemma \ref{lemma:G}.(ii), $G'_{11}$ is a block of a Vandermonde matrix and therefore is invertible. By Lemma \ref{lemma:G}.(iii), $G'_{21} = 0$. Denote  $\bm{c}_I - \bm{m} G_{\text{top}, I}$ by $(u_1, ..., u_{(k+r)z})$, then the above two facts together with (\ref{eq:kk'rs}) imply that
\begin{align}\label{eq:k}
\bm{k} =  (u_1, ..., u_{kz}) G'^{-1}_{11}
\end{align}
Therefore $\bm{k}$ is a deterministic function of $\bm{m}$ and $\bm{c}_I$. It follows from (\ref{eq:kk'rs}) that
\begin{align*}
  \bm{k}'G'_{22} = (u_{kz+1}, ..., u_{(k+r)z}) - \bm{k}G'_{12}.
\end{align*}
By Lemma $\ref{lemma:G}$.(iv), $G'_{22}$ is invertible and therefore
\begin{align}
  \bm{k}' = \left( (u_{kz+1}, ..., u_{(k+r)z}) - \bm{k}G'_{12} \right) G'^{-1}_{22}.
\end{align}
This shows that $\bm{k}'$ is a deterministic function of $\bm{k}$, $\bm{c}_I$ and $\bm{m}$, and so 
\begin{align}\label{eq:kk'}
H( \bm{k}, \bm{k}' |\bm{c}_I, \bm{m})=0.
\end{align}
It then follows that,
\begin{align}\label{eq:secfinalstep}
H(\bm{m}) - H(\bm{m}|\bm{c}_I)  & = I(\bm{m} ; \bm{c}_I) \nonumber\\
& = H(\bm{c}_I) - H(\bm{c}_I | \bm{m})\nonumber\\
& \stackrel{(f)}{\le} z - H(\bm{c}_I | \bm{m}) \nonumber\\ 
& \stackrel{(g)}{=}  z - H(\bm{c}_I | \bm{m}) + H(\bm{c}_I | \bm{m}, \bm{k}, \bm{k}')\nonumber\\
& = z - I(\bm{c}_I ;  \bm{k}, \bm{k}' | \bm{m})\nonumber\\
& = z - H( \bm{k}, \bm{k}' | \bm{m} ) + H(\bm{k}, \bm{k}' | \bm{c}_I, \bm{m})\nonumber\\
& \stackrel{(h)}{=} z - H(\bm{k},\bm{k}'|\bm{m})\nonumber\\
& \stackrel{(i)}{=} z - H(\bm{k}, \bm{k}')\nonumber\\
& \stackrel{(j)}{=} z-z=0,
\end{align}
where (f) is due to $|I| = z$; (g) is due to the fact that $\bm{c}_I$ is a function of $\bm{m}$, $\bm{k}$ and $\bm{k}'$; (h) is due to (\ref{eq:kk'});  (i) is due to the fact that $\bm{k}, \bm{k}'$ are independent of $\bm{m}$; and (j) follows from the fact that $\bm{k}, \bm{k}'$ are uniformly distributed.
Therefore $H(\bm{m}) = H(\bm{m}|\bm{c}_I)$ and the security requirement is met. This completes the proof that the encoding scheme is a valid secret sharing scheme.
\end{proof}
Theorem \ref{th:con1} shows that the proposed secret sharing scheme is optimal in terms of storage  usage and is optimal in terms of best-case (i.e., $|I| = n$) communication overhead. Compared to the scheme in the previous section, this scheme has advantages in terms of implementation and error correction because decoding the scheme is equivalent to decoding standard Reed-Solomon codes. The scheme also provides a stronger level of reliability in the sense that it allows decoding even if more than $r$ shares are partially erased. Similar to previous discussion, the scheme achieves the optimal number of symbol-reads from disks when $|I|=n$.
Finally, in the scheme all operations are performed over the field $\mathbb{F}_q$, where $q>n(k+r)$.  This requirement on the field size can be relaxed in the following simple way. Let $\beta$ be the greatest common divisor of $k$ and $r$,  then instead of choosing $\mathcal{Q}$ to be $\mathbb{F}_q^{k+r}$, we can let $\mathcal{Q} = \mathbb{F}_q^{\frac{k}{\beta} + \frac{r}{\beta}}$, $\bm{m} = (m_1, ..., m_{\frac{k(k+r)}{\beta}})$, $\bm{k} = (k_1, ..., k_{\frac{kz}{\beta}})$ and $\bm{k}' = (k'_1, ..., k'_{\frac{rz}{\beta}})$. The resulting scheme is a rate-optimal $(n,k,r,z)_{\mathcal{Q}}$ secret sharing scheme with the same communication overhead function as the original scheme. For this modified construction, it is sufficient to choose any field size $q>n\frac{k+r}{\beta}$. 
\section{Secret Sharing Schemes from Random Codes}\label{sec:rdm}
In this section we describe a rate-optimal (perfect or ramp) secret sharing scheme based on random linear codes that achieves the optimal decoding bandwidth universally. The scheme meets the secrecy requirement deterministically, and meets the reliability requirement with high probability as the field size grows.

\subsection{Encoding}\label{sec:enc2}
Let $k=n-r-z$, $q$ be a prime power, and let $N$ be the least common multiple of $\{n-z-r, , n-z-r+1, ..., n-z\}$. Set 
$\mathcal{Q} = \mathbb{F}_q^{N(k+r)}$. Therefore each share of the secret is a length $N(k+r)$ vector over $\mathbb{F}_q$. For $j=1,...,n$, denote the $j$-th share by $c_j = (c_{1,j}, ..., c_{N(k+r),j})$, where $c_{i,j} \in \mathbb{F}_q$. The secret $\bm{m}$ consists of $k$ symbols over $\mathcal{Q}$ and is regarded as a length-$Nk(k+r)$ vector over $\mathbb{F}_q$, denoted by $(m_1,...,m_{Nk(k+r)})$. The encoder generates uniformly distributed random vectors $\bm{k}=(k_1, ..., k_{Nkz}) \in \mathbb{F}_q^{Nkz}$ and $\bm{k}'=(k'_1,...,k'_{Nrz}) \in \mathbb{F}_q^{Nrz}$, independently from $\bm{m}$. The encoding scheme is described by a set of $N(kn+rz) \times n$ encoding matrices $G_i, i=1,...,N(k+r)$ over $\mathbb{F}_q$, such that
\begin{align}\label{eq:enc2}
(c_{i,1}, ..., c_{i,n}) = (\bm{m} \ \bm{k} \ \bm{k}')G_{i}, \ \ \ \ \ i=1, ..., N(k+r).
\end{align}
Intuitively, if the $c_{u,v}$'s are arranged into a matrix, then $G_i$ is the encoding matrix for the $i$-th row. We next describe the construction of the $G_i$ matrices. For $i=1, ..., Nk$, let the first $Nk(k+r)$ rows of $G_i$ be a random matrix, let the next $Nkz$ rows of $G_i$ be a Vandermonde matrix, and let the remaining $Nrz$ rows of $G_i$ be zero. Formally, for $i=1, ...,Nk$,
\begin{align}\label{eq:Gi}
G_i = \left( \begin{array}{c} R_i \\ V_i \\ \bm{0} \end{array} \right),
\end{align}
where $R_i \in \mathbb{F}_q ^{Nk(k+r) \times n}$ is a random matrix with entries i.i.d. uniformly distributed over $\mathbb{F}_q$, and $V_i \in \mathbb{F}_q^{Nkz \times n}$ is a Vandermonde matrix, i.e.,  the $(u,v)$-th entry of $V_i$ equals $\alpha_{v,i}^{u-1}$. Here $\alpha_{v,i}$ are distinct non-zero elements of $\mathbb{F}_q$, for $i=1,...,N(k+r)$, and $v=1,...,n$.

For $i=1 ,..., Nr$, let the first $Nkn + (i-1)z$ rows of $G_{Nk+i}$ be a random matrix, let the next $z$ rows of $G_{Nk+i}$ be a Vandermonde matrix, and let the remaining $(Nr-i)z$ rows of $G_{Nk+i}$ be zero. Formally, for $i=1,...,Nr$,
\begin{align}\label{eq:Gki}
G_{Nk+i} = \left( \begin{array}{c} R_{Nk+i} \\ V_{Nk+i} \\ \bm{0} \end{array} \right),
\end{align}
where $R_{Nk+i} \in \mathbb{F}_q ^{(Nkn+(i-1)z) \times n}$ is a random matrix with entries i.i.d. uniformly distributed over $\mathbb{F}_q$, and $V_{Nk+i} \in \mathbb{F}_q^{z \times n}$ is a Vandermonde matrix, i.e.,  the $(u,v)$-th entry of $V_{Nk+i}$ equals $\alpha_{v,Nk+i}^{u-1}$. This completes the encoding scheme. The structure of the whole encoding matrix $(G_1, ..., G_{N(k+r)})$ is illustrated in Figure \ref{fig:rndG}.
 \begin{figure}[htb]
 \centering
 \includegraphics[width=0.7\textwidth]{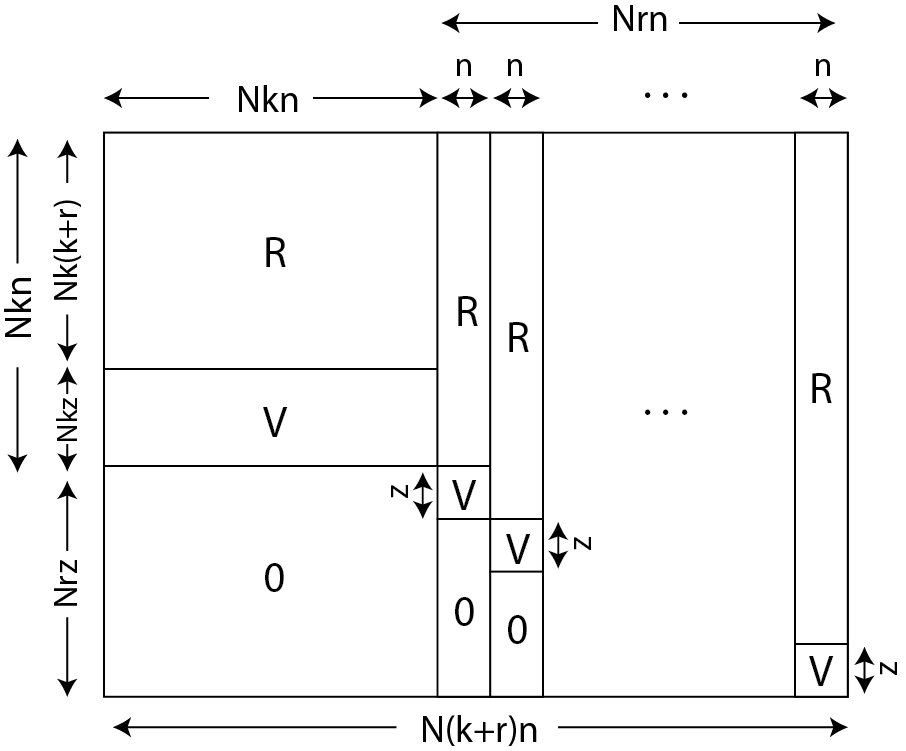}
                 \caption{Blockwise structure of the matrix $(G_1, ..., G_{N(k+r)})$. Blocks of random matrices are labelled by {\fontfamily{qhv}\selectfont R}, blocks of Vandermonde matrices are labelled by {\fontfamily{qhv}\selectfont V}, and blocks of zero matrices are labelled by {\fontfamily{qhv}\selectfont 0}. }
                \label{fig:rndG}
\end{figure}

The following result shows that the scheme meets the security requirement deterministically, due to the Vandermonde matrices embedded in the $G_i$'s. 
\begin{theorem}\label{th:sec2}
The encoding scheme constructed in this section is secure, i.e., $H(\bm{m}|\bm{c}_I) = H(\bm{m})$, for all $I$ such that $|I|=z$.
\end{theorem}
\begin{proof}
Consider any $I$ such that $|I|=z$. As in Theorem \ref{th:con1}, we first show that $H(\bm{k},\bm{k}'|\bm{c}_I, \bm{m})=0$. Denote by $G_{i,I}$, $R_{i,I}$ and $V_{i,I}$ the submatrix formed by the set of columns in $I$ of $G_i$, $R_i$ and $V_i$, respectively.  Let
\begin{align*}
G_{1 \to Nk, I} &= (G_{1,I} \ ...\ G_{Nk,I})\\
R_{1 \to Nk, I} &= (R_{1,I} \ ... \ R_{Nk,I})\\
V_{1 \to Nk, I} &= (V_{1,I} \ ... \ V_{Nk,I}).
\end{align*}
Denote for short $\mathfrak{c}_{i} = (c_{i,j})_{j \in I}$, then by (\ref{eq:enc2}), it follows that
\begin{align*}
 (\bm{m} \ \bm{k} \ \bm{k}')G_{1 \to Nk, I} & =  (\bm{m}\ \bm{k} \ \bm{k}') \left( \begin{array}{c} R_{1 \to Nk, I} \\ V_{1 \to Nk, I} \\ \bm{0} \end{array} \right) \\
 & = (\mathfrak{c}_1 , ..., \mathfrak{c}_{Nk}).
\end{align*}
Notice that $V_{1 \to Nk, I}$ is a $Nkz \times Nkz$ square Vandermonde matrix. Therefore it is invertible and 
\begin{align*}
\bm{k} = \left( (\mathfrak{c}_1 , ..., \mathfrak{c}_{Nk}) - \bm{m}R_{1 \to Nk, I} \right) V^{-1}_{1 \to Nk, I}.
\end{align*}
Hence $H(\bm{k} | \bm{c}_I, \bm{m} )=0$. Then by (\ref{eq:enc2}), it follows that
\begin{align*}
(\bm{m} \ \bm{k}\  |\  k'_1, ..., k'_z|\  k'_{z+1}, ..., k'_{Nrz} ) \left( \begin{array}{c} R_{Nk+1,I}  \\ \hline V_{Nk+1, I} \\ \hline \bm{0} \end{array} \right) = \mathfrak{c}_{Nk+1}.
\end{align*}
Notice that $V_{Nk+1, I}$ is a $z \times z$ square Vandermonde matrix. Therefore it is invertible and 
\begin{align*}
(k'_1, ..., k'_z) = \left( \mathfrak{c}_{Nk+1} - (\bm{m} \ \bm{k}) R_{Nk+1, I } \right) V^{-1}_{Nk+1,I}.
\end{align*}
Hence $H(k'_1, ..., k'_z | \bm{k}, \bm{c}_I, \bm{m}) = 0$. Similarly, we can show that for $i=1,...,Nr$ $$H(k'_{(i-1)z+1}, ..., k'_{iz} | k'_{1}, ..., k'_{(i-1)z}, \bm{k}, \bm{c}_I, \bm{m}) = 0.$$ 
Therefore by the chain rule,
\begin{align}\label{eq:sec2kk'}
H(\bm{k},\bm{k}' | \bm{c}_I, \bm{m}) & = H(\bm{k} | \bm{c}_I, \bm{m}) + \sum_{i=1}^{Nr} H(k'_{(i-1)z+1}, ..., k'_{iz} | k'_{1}, ..., k'_{(i-1)z}, \bm{k}, \bm{c}_I, \bm{m})\nonumber\\
& = 0.
\end{align}
Provided that $(\ref{eq:sec2kk'})$ is true, we can then follow exactly the same argument as (\ref{eq:secfinalstep}) in the proof of Theorem \ref{th:con1}, to show that $H(\bm{m}|\bm{c}_I) = H(\bm{m})$. This completes the proof.
\end{proof}
\subsection{Decoding}\label{sec:dec2}
We describe the decoding scheme for any $I$ such that $|I| \ge n-r$. Let  
\begin{align}\label{eq:defd}
d \triangleq  \frac{Nk(n-|I|)}{|I|-z},
\end{align}
then note that $d$ is an integer because $|I|-z$ divides $N$ and that $d$ is the solution to the equation  $(Nk+d)|I| = Nkn+dz$.
In order to decode, it suffices to read and communicate the first $Nk+d$ symbols over $\mathbb{F}_q$ from each available share. Intuitively, by reading the first $Nk+d$ symbols from each available share, we have a system of $(Nk+d)|I|$ equations. On the other hand, the variables involved in these equations are $m_1, ..., m_{Nk(k+r)}$, $k_1, ..., k_{Nkz}$ and $k'_1, ..., k'_{dz}$, i.e., the total number of variables is $Nkn+dz$. Because $d$ is the solution to $(Nk+d)|I| = Nkn+dz$, the number of equations in the system equals the number of variables, and is uniquely solvable if the equations are linearly independent.

Formally, let $\mathcal{S}_i = \mathbb{F}_q^{Nk+d}$, and let $E_i(c_i) = (c_{1,i}, ..., c_{Nk+d,i})$. Denote for short that $\mathfrak{c}_{i} = (c_{i,j})_{j \in I}$, and denote  the submatrix formed by the set of columns in $I$ of $G_i$ by $G_{i,I}$. Then it follows from (\ref{eq:enc2}) that,
\begin{align*}
\bm{e}_I = ( \mathfrak{c}_1 , ...,  \mathfrak{c}_{Nk+d} )= (\bm{m} \ \bm{k}\ \bm{k}') ( G_{1,I}, ..., G_{Nk+d,I} ).
\end{align*}
By construction (\ref{eq:Gi}) and (\ref{eq:Gki}), the last $(Nr-d)z$ rows of the matrices $ G_{1,I}, ..., G_{Nk+d,I} $ are all zeros. Therefore we may delete the last $(Nr-d)z$ rows from $ G_{1,I}, ..., G_{Nk+d,I}$ and denote by $( {G}^*_{1,I}, ..., {G}^*_{Nk+d,I} )$ the corresponding trimmed matrix. It then follows that,
\begin{align*}
\bm{e}_I = ( \mathfrak{c}_1 , ...,  \mathfrak{c}_{Nk+d} )= (\bm{m} \ \bm{k}\ k'_1, ..., k'_{dz}) ( {G}^*_{1,I}, ..., {G}^*_{Nk+d,I} ).
\end{align*}
It is now evident that if the matrix $(G^*_{1,I}, ..., G^*_{Nk+d,I})$ has full row rank, then it is right invertible and  the secret can be recovered. The following result shows that the matrix indeed has full row rank with high probability. 
\begin{theorem}\label{th:dec2}
For any $I$ such that $|I| \ge n-r$, $(G^*_{1,I}, ..., G^*_{Nk+d,I})$ has full row rank with probability at least $1-\frac{1}{q-1}$, over the distribution of the random matrices $R_1, ..., R_{Nk+d}$.
\end{theorem}
\begin{proof}
Note that $(G^*_{1,I}, ..., G^*_{Nk+d,I})$ has size $(Nkn+dz) \times (Nk+d)|I|$. By the definition of $d$, it follows that $Nkn+dz = (Nk+d)|I|$ and therefore the matrix is square. Hence it suffices to show the matrix has full column rank and in the following we show the columns of the matrix are linearly independent with high probability.

The first $Nkz$ columns of $(G^*_{1,I}, ..., G^*_{Nk,I})$ are linearly independent because by (\ref{eq:Gi}), the $(Nkn-Nkz+1)$-th row to the $Nkn$-th row form a Vandermonde matrix. Denote the $i$-th column of $(G^*_{1,I}, ..., G^*_{Nk+d,I})$ by $\bm{g}_i$. We first study the probability that the $\bm{g}_{Nkz+i}$ is in the linear span of all the previous columns, i.e., span$[\bm{g}_1, ..., \bm{g}_{Nkz+i-1}]$, for $i=1,...,Nk(|I|-z)$. 
Consider the sum of vectors $ \bm{g}^*  = \sum_{l=1}^{NKz+i-1} \gamma_l \bm{g}_l$. Fixing $\gamma_1, ..., \gamma_{i-1}$ to be arbitrary values in  $\mathbb{F}_q$, then there is a unique tuple $(\gamma_{l})_{l=i}^{NKz+i-1}$ such that $\bm{g}^*$ agrees with $\bm{g}_{NKz+i}$ in the $(Nkn-Nkz+1)$-th to $Nkn$-th entries.
Therefore there are $q^{i-1}$ different ways to linearly combine $\bm{g}_1, ..., \bm{g}_{Nkz+i-1}$, such that in the resulting sum vector, the $(Nkn-Nkz+1)$-th to $Nkn$-th entries are equal to the corresponding entries of $\bm{g}_{Nkz+i}$. Because the first $Nk(n-z)$ entries of $\bm{g}_{Nkz+i}$ are i.i.d. uniformly distributed, it follows that
\begin{align}\label{eq:left}
\Pr \{ \bm{g}_{Nkz+i} \in \text{span}[\bm{g}_1, ..., \bm{g}_{Nkz+i-1}] \} \le \frac{q^{i-1}}{q^{Nk(n-z)}}, \ \ \ \ i=1,...,Nk(|I|-z)
\end{align}

We next study the probability that $\bm{g}_{Nk|I| + i}$ is in span[$\bm{g}_1, ..., \bm{g}_{Nk|I|+i-1}$]. Consider arbitrary $Nk \le j \le Nk+d-1$. By construction (\ref{eq:Gki}), $\bm{g}_{j|I| + i} \notin \text{span}[\bm{g}_1, ...,  \bm{g}_{j|I| + i -1}]$, for $1 \le i \le z$, due to the Vandermonde matrix $V_{Nk+j}$. Now consider $\bm{g}_{j|I| + i}$ with $z+1 \le i \le |I|$. There are $q^{j|I|+i-z-1}$ different ways to linearly combine $\bm{g}_1, ...,  \bm{g}_{j|I| + i -1}$, such that in the resulting sum vector, the $(Nkn + (j-Nk)z+1)$-th to $(Nkn + (j-Nk)z+z)$-th entries are equal to the corresponding entries in $\bm{g}_{j|I|+i}$. Note that the first $Nkn+(j-Nk)z$ entries of $\bm{g}_{j|I| + i}$ are i.i.d. uniformly distributed. Therefore, for $Nk \le j \le Nk+d-1$ and $z+1 \le i \le |I|$, it follows that
\begin{align}\label{eq:right}
\Pr \{ \bm{g}_{j|I|+i} \in \text{span}[\bm{g}_1, ..., \bm{g}_{j|I|+i-1}] \} \le \frac{q^{j|I|+i-z-1}}{q^{Nkn+(j-Nk)z}}.
\end{align}
Hence, by the union bound
\begin{align}
\Pr \{(G^*_{1,I}, ..., G^*_{Nk+d,I}) \text{ singular}\} & \nonumber\\  
& \hspace{-30mm} \le \sum_{i=1}^{Nk(|I|-z)} \Pr \{ \bm{g}_{Nkz+i} \text{ l.d.}^2\} + \sum_{j=Nk}^{Nk+d-1}\sum_{i=z+1}^{|I|} \Pr \{ \bm{g}_{j|I|+i} \text{ l.d.}\} \nonumber\\
 & \hspace{-30mm} \stackrel{(k)}{\le}   \sum_{i=1}^{Nk(|I|-z)} \frac{q^{i-1}}{q^{Nk(n-z)}} + \sum_{j=Nk}^{Nk+d-1} \sum_{i=z+1}^{|I|} \frac{q^{j|I| + i -z -1}}{q^{Nkn+(j-Nk)z}}  \nonumber\\
 & \hspace{-30mm} = \sum_{i=Nk(z-n)}^{Nk(|I|-n) -1 } q^{i}  +  \sum_{j=Nk}^{Nk+d-1} \sum_{i=z+1}^{|I|} \frac{q^{j|I| + i -z -1}}{q^{Nkn+(j-Nk)z}}  \nonumber\\
 & \hspace{-30mm} = \sum_{i=Nk(z-n)}^{Nk(|I|-n) -1 } q^{i}  + \sum_{j=Nk}^{Nk+d-1} \sum_{i=0}^{|I|-z-1} q^{j(|I|-z)+i+Nk(z-n)}  \nonumber\\
  & \hspace{-30mm} = \sum_{i=Nk(z-n)}^{Nk(|I|-n) -1 } q^{i}  + q^{Nk(z-n)} \sum_{j=Nk}^{Nk+d-1} \sum_{i=0}^{|I|-z-1} q^{j(|I|-z)+i}  \nonumber\\
  & \hspace{-30mm} = \sum_{i=Nk(z-n)}^{Nk(|I|-n) -1 } q^{i}  + q^{Nk(z-n)} \sum_{i=Nk(|I|-z)}^{(Nk+d)(|I|-z)-1}  q^i\nonumber\\
  & \hspace{-30mm} = \sum_{i=Nk(z-n)}^{Nk(|I|-n) -1 } q^{i}  + \sum_{i=Nk(|I|-n)}^{(Nk+d)|I|-Nkn-dz-1}  q^i\nonumber\\
 &  \hspace{-30mm} \stackrel{(l)}{=} \sum_{i=Nk(z-n)}^{Nk(|I|-n) -1 } q^{i}  + \sum_{i=Nk(|I|-n)}^{-1} q^i   \nonumber\\
 & \hspace{-30mm}  < \sum_{i=-\infty}^{-1} q^i = \frac{1}{q-1},
 \end{align}
\addtocounter{footnote}{1}
\footnotetext{Linearly dependent on the set of columns to the left.}
where ($k$) is due to (\ref{eq:left}) and (\ref{eq:right}), and ($l$) is due to (\ref{eq:defd}). This completes the proof.
\end{proof}
The following result summarizes the properties of scheme.
\begin{corollary}
The encoding scheme constructed in Section \ref{sec:enc2} is a rate-optimal $(n,k,r,z)$ secret sharing scheme with high probability. Specifically, the scheme meets the security requirement deterministically, and meets the reliability requirement with probability at least $1 - \frac{2^n}{q-1}$, over the distribution of the random matrices $R_1, ..., R_{Nk+d}$. The scheme achieves the optimal decoding bandwidth when $d$ nodes participate in decoding, universally for all $n-r \le d \le n$. 
\end{corollary}
\begin{proof}
The scheme achieves capacity because $k=n-r-z$. By Theorem  \ref{th:sec2}, the scheme meets the security requirement. By Theorem \ref{th:dec2} and the union bound, the scheme meets the reliability requirement with probability at least $1-\sum_{i=0}^{r} {n \choose i} \frac{1}{q-1} \ge 1 - \frac{2^n}{q-1}$. 

Consider any $I$ such that $|I|\ge n-r$. In order to decode, a number of $(Nk+d)|I|$ symbols over $\mathbb{F}_q$ are communicated. Therefore the communication overhead is
\begin{align*}
\text{CO}(I)  &= \frac{(Nk+d)|I| - Nk(k+r)}{N(k+r)}\\
& = \frac{Nkn + dz - Nk(k+r)}{N(k+r)}\\
& = \frac{Nkz + dz}{N(k+r)} = \frac{z(Nk + \frac{Nk(n - |I|)}{|I| - z})}{ N(k+r)} \\
& = \frac{Nk(n-z)z}{(|I|-z)N(k+r)} = \frac{kz}{|I| - z},
\end{align*}
which achieves equality in (\ref{eq:upbnd}).
\end{proof}


\section{Conclusions}
In this paper we study the communication efficiency of secret sharing schemes in decoding. We prove an information-theoretic lower bound on the amount of information to be communicated during decoding, and show that the decoding bandwidth decreases as $d$, the number of nodes that participate in decoding, increases. We prove that the bound is uniformly tight by designing a secret sharing scheme that achieves the optimal decoding bandwidth universally for all valid $d$. The scheme is simple and is efficient in both space and computation. We construct another secret sharing scheme  that achieves the optimal decoding bandwidth when all nodes are available. The scheme has an advantage in implementation because its codewords form the Reed-Solomon codes. In the application of distributed storage, the proposed communication efficient secret sharing schemes also improve disk access efficiency. There are a number of interesting open problems: 1) in the application of distributed storage, how can one construct codes that are communication efficient in terms of both decoding and repair? 2) how to generalize the results to other (non-threshold) access structures? and 3) is it possible to extend the schemes and ideas in the paper to improve the communication efficiency of other secure protocols that use secret sharing schemes as building blocks?

\section*{Acknowledgment}
We thank an anonymous reviewer for pointing us to the paper by Wang and Wong \cite{Wang08}.

%

\ifCLASSOPTIONcaptionsoff
  \newpage
\fi



%
\bibliographystyle{IEEEtranS}
\bibliography{sha}

\end{document}